\documentclass[11pt]{article}

\usepackage{enumerate}

\usepackage{blurb_style}
\input{EE.sty}

\usepackage{setspace}

\usepackage{graphicx}
\usepackage{epstopdf}
\usepackage{natbib}

\usepackage{bm}
\usepackage{epigraph} 
\usepackage[table]{xcolor}
\usepackage{multirow}

\usepackage{listings}
\DeclareMathSymbol{\bS}{\mathbin}{AMSb}{"53}
\newcommand{\sch}{\bS}
\newcommand\schimp{\tilde{\sch}}
\newcommand\Yobs{Y^{\text{obs}}}

\newcommand\yobsV{\textbf{y}^{\text{obs}}}
\newcommand\Yimp{\tilde{Y}}

\newcommand\Wobs{W^{\text{obs}}}
\newcommand\wobs{w^{\text{obs}}}

\newcommand\wobsV{\textbf{w}^{\text{obs}}}
\newcommand\pimp{\tilde{p}}
\newcommand\tobs{{t}^{\text{obs}}}

\newcommand\taumax{\tau^{\vee}}
\newcommand\vtaumax{\vtau^{\vee}}
\newcommand\taumin{\tau^{\wedge}}
\newcommand\vtaumin{\vtau^{\wedge}}
\newcommand\sharpH{\rH_{\vtau^{0}}}
\newcommand\nonsupH{\rH_{\vtaumax}}
\newcommand\noninfH{\rH_{\vtaumin}}
\newcommand\boundedH{\nonsupH}

\usepackage{color}

\newcommand{\poscite}[1]{\citeauthor{#1}'s (\citeyear{#1})} 

\usepackage{hyperref}

\begin{document}

\author{Devin Caughey\\MIT \and Allan Dafoe\\Yale \and Luke Miratrix\\Harvard}
\title{%
  Beyond the Sharp Null: \\
  Randomization Inference, Bounded Null Hypotheses,\\
  and Confidence Intervals for Maximum Effects%
\thanks{Authors can be reached at devin.caughey@gmail.com, allandafoe@gmail.com, and lmiratrix@g.harvard.edu. For helpful input we thank Peter Aronow, Jake Bowers, Joanna Dafoe, Danny Hidalgo, Greg Huber, Kosuke Imai, Luke Keele, Kelly Rader, Brandon Stewart, and seminar participants at PolMeth 2016. We also thank the late Natasha Chichilnisky-Heal for posing the question that stimulated some of our thinking on this paper. 
}
}

\singlespacing

\maketitle
\begin{abstract}
\noindent  Fisherian randomization inference is often dismissed as testing an uninteresting and implausible hypothesis: the sharp null of no effects whatsoever. We show that this view is overly narrow. Many randomization tests are also valid under a more general ``bounded'' null hypothesis under which all effects are weakly negative (or positive), thus accommodating heterogenous effects.  By inverting such tests we can form one-sided confidence intervals for the maximum (or minimum) effect. These properties hold for all effect-increasing test statistics, which include both common statistics such as the mean difference and uncommon ones such as Stephenson rank statistics. The latter's sensitivity to extreme effects permits detection of positive effects even when the average effect is negative. We argue that bounded nulls are often of substantive or theoretical interest, and illustrate with two applications: testing monotonicity in an IV analysis and inferring effect sizes in a small randomized experiment.
\end{abstract}
\bigskip

\doublespacing

\setlength{\epigraphwidth}{0.6\textwidth}


%
%


\section{Motivation}
\label{sec:motivation}


Randomization inference (RI), also known as permutation inference, is a general statistical framework for learning about treatment effects. 
RI originated with R. A. \citet{Fisher35a}, who demonstrated that if treatment is randomly assigned to units, the hypothesis that no unit was affected by treatment---the so-called ``sharp null of no effects''---can be tested exactly,  with no further assumptions, by comparing an observed test statistic with its distribution across alternative realizations of treatment assignment.\footnote{
  More precisely, a randomization test of the sharp null of no effects is a \textit{valid} test, meaning that its whose true false-rejection rate is no greater than (as opposed to exactly equal to) its significance level $\alpha$.
}
 Moreover, by testing a sequence of hypotheses, randomization tests can also be used to create exact nonparametric confidence intervals (CIs) for treatment effects \citep{Lehmann63a}. RI thus provides a unified framework of 
 statistical inference that requires neither  parametric assumptions about the data-generating distribution nor asymptotic approximations that may be unreliable in small samples  (\citet{Rosenbaum02a}, but for applications in political science, see \citet{HoImai06a, Keele12a, Bowers13a, Glynn14a, CaugheyDafoeSeawright17NPC}).\footnote{
  RI's finite-sample validity provides an important contrast with nonparametric methods that rely on asymptotic approximations, such as  average treatment effect (ATE) estimation. For example, Welch's unequal-variances $t$ test, which is an asymptotically conservative nonparametric test of the ATE \citep{SamiiAronow12a}, is often considered highly robust, even in moderately sized samples of $n = 30$. But even the $t$ test can be quite inaccurate if the sample sizes differ between treatment and control and the response distributions are skewed \citep[e.g.,][372]{Hesterberg15a}. Appendix \ref{sec:performance-t-test} illustrates this with the example of a skewed beta distribution and sample sizes of $n_1=1000$ and $n_2=30$. Under these conditions, the $t$ test with $\alpha=0.01$ falsely rejects the null of mean equality over 10\% of the time. By contrast, Fisher's difference-of-means permutation test maintains exact coverage under these conditions.
}

Despite its virtues, RI has also been subject to trenchant critiques. One stems from the claim that RI can only be used to test \textit{sharp} null hypotheses---that is, hypotheses that stipulate the treatment effect of each unit. Since sharp hypotheses often have little prior plausibility, rejecting them arguably conveys little information of scientific or substantive value. A second line of criticism, related to the first, is that extending RI from hypothesis testing to interval estimation requires assumptions that are arguably as strong as those of its parametric and large-sample competitors. In particular, deriving interpretable CIs typically requires the assumption that all treatment effects are constant (that is, do not vary across units), which is often highly implausible, especially in the social sciences. For these reasons, social-science methodologists and applied researchers have typically viewed RI as a highly limited tool, useful primarily for assessing whether treatment had any effect at all.

These critiques have substantial merit, but we argue that they are based on an overly restrictive and pessimistic view of RI. We show that many randomization tests can in fact be interpreted as conservative tests of a ``weak'' null hypothesis---that is, one under which unit-level effects are heterogeneous and need not be precisely stipulated. Specifically, we prove that for a broad class of randomization tests, one-sided rejection of the sharp null that all treatment effects $\tau_i$ equal some constant $\tau^0$ also implies rejection of \emph{any} null hypothesis under which the $\tau_i$ are bounded on one side by $\tau^0$. This means that RI can be used test the \textit{non-superiority} null that all effects are less than or equal to $\tau^0$ ($\tau_i \le \tau^0 \; \forall i$) or, alternatively, the \textit{non-inferiority} null $\tau_i \ge \tau^0 \; \forall i$. We refer to these as \textit{bounded}  null hypotheses. 

We show that any randomization test of a given sharp null will be valid under the corresponding bounded null, as long as it employs an \textit{effect-increasing} test statistic. Loosely speaking, an effect-increasing test statistic is one that increases in value as the treated responses increase and the control responses decrease.\footnote{
  We borrow the term \textit{effect-increasing} test statistics from \citet[37--8]{Rosenbaum02a}, who discusses them in the context of power against particular alternatives under the sharp null. Our definitions of \textit{effect-increasing }differ slightly, in that Rosenbaum defines it in terms of observed outcomes and we do so in terms of potential outcomes.
} Although some test statitics, such as Student's $t$ statistic, are not effect-increasing, many commonly used statistics are, including the difference of means and the Wilcoxon rank sum. 

Compared to the sharp null of no effects, bounded null hypotheses are 
much more plausible a priori.
They are also normatively important. In order for a treatment or intervention to be Pareto improving, for example, it must make at least one person better off while hurting no one. Thus, rejecting the bounded null hypothesis that all effects are greater than or equal to zero implies rejection of the hypothesis that a treatment is Pareto improving. At a more technical level, bounded null hypotheses are also invoked as assumptions by a variety statistical methods. Instrumental variables estimation, for example, is typically conducted under a monotonicity assumption that the instrument has non-negative or non-positive effects on the treatment.
There are thus many situations where testing a bounded null hypothesis is of theoretical or practical interest.

In addition to hypothesis testing, our results also provide a basis for interval estimation. 
Rejecting $\tau_i = \tau^0\ \forall i$ with a one-sided randomization test implies rejection of the bounded hypothesis $\tau_i \le \tau^0\ \forall i$. 
Thus, by inverting a sequence of tests for different values of $\tau^0$, it is possible to form a one-sided confidence interval for the maximum treatment effect (analogous logic applies for the minimum effect). 
These intervals give a sense of the size of the effects for most impacted units.
They can also be thought of as confidence intervals stating where at least some of the individual causal effects lie, and the more homogeneous the effects are thought to be, the more effects these intervals would contain. 
In particular, when using test statistics sensitive to the central tendency, inferences on the upper bound will generally coincide closely with inferences for the ATE. As we demonstrate in our example applications, however, the inferences for the mean and the maximum can diverge substantially if one uses a test statistic sensitive to extreme treatment effects. Indeed, even if the ATE is significantly negative, it is nevertheless possible to conclude that treatment had at least one positive effect of a given magnitude.

In sum, we offer a novel reinterpretation of randomization inference. We show that RI can do much more than merely assess whether treatment had any effect at all. Not only are randomization tests valid under a much less restrictive null hypothesis than is commonly understood, but, unlike ATE estimators, they can be used to construct interval estimates for the maximum or minimum treatment effect, a quantity that is often of normative or theoretical interest. 
 All this is possible without asymptotic approximations or any additional assumptions beyond random assignment and SUTVA \citep{rubin1980randomization}. 

The remainder of the paper is organized as follow. After first reviewing existing perspectives on RI, we illustrate the intuition behind randomization tests' validity under the bounded null. We then prove this result formally for the class of effect-increasing test statistics, and show that this result implies that RI can be used to derive exact confidence intervals for the maximum or minimum treatment effect. We follow with a discussion of effect-increasing statistics, noting that more precise confidence intervals can be obtained by using statistics sensitive to extreme effects. After discussing the practical and theoretical relevance of bounded null hypotheses, we illustrate its value with two example applications. We then offer brief conclusions.

\section{Current Perspectives on Randomization Inference}
\label{sec:persp-rand-infer}

As noted above, critiques of RI generally focus on two main limitations. The first is that randomization tests are valid only under a specific sharp null hypothesis, such as that no unit has any treatment effect or that all units have the same effect. By contrast, large-sample nonparametric procedures such as the $t$ test are valid under a so-called ``weak'' null hypothesis that specifies the value of some function of the treatment effects, typically the ATE. The sharp null of no effects implies the weak null of mean equality, but the converse is not true. Indeed, without further assumptions, permutation tests are not valid tests of a weak null, even asymptotically \citep{Romano90a}. For example, if treatment groups have different sizes and variances, rejecting the sharp null of no effects with a difference-of-means permutation test does not imply that the weak null of no average effect can also be validly rejected. For these reasons, many scholars follow \citet[173]{Neyman35a} in regarding sharp nulls as ``uninteresting and academic.'' \citet{Gelman11a}, for example, argues that ``the so-called Fisher exact test almost never makes sense, as it's a test of an uninteresting hypothesis of exactly zero effects (or, worse, effects that are nonzero but are identical across all units).'' The current consensus among political methodologists is aptly summarized by \citet[330]{Keele15a}, who states that because it does not easily ``accommodate heterogeneous responses to treatment,'' the sharp null ``would seem to be a very restrictive null hypothesis.'' 

The crux of this critique, however, concerns not hypothesis testing per se, but rather inference regarding causal quantities of interest. From this perspective, the real value of Neyman-style causal inference is not that it permits a given weak null to be tested, but that it provides a basis for interval estimation of some summary of the unit-level effects, namely their average.\footnote{
  By \textit{interval estimation} we mean the calculation of a numerical range to summarize the probable value of an unknown quantity. By contrast, point estimation summarizes the probable value with a single ``best guess.'' 
} Confidence intervals for the ATE indicate the range of plausible values of this parameter that are consistent with the data. Although CIs can also be constructed within an RI framework, doing so typically requires the assumption that effects are constant across units. This constant-effects assumption is rightly viewed as quite strong, especially in social science, ``where treatment effect heterogeneity is a rule rather than an exception.''\footnote{
  This quotation is from Kosuke Imai's lecture notes on permutation tests (\citeyear[7]{Imai13a}).
} Without some version of this assumption, however, RI as currently understood provides no basis for drawing exact inferences about the magnitude of treatment effects.

Defenders of randomization inference have responded to these critiques in various ways. Some, while largely accepting the critiques of the sharp null, argue that RI is nevertheless useful for assessing whether treatment had any effect at all, as a preliminary step to determine whether further analysis is warranted \citep[e.g.,][]{Imbens15a}. 
An alternative proposal, advanced by \citet{Chung13a}, is to employ ``studentized'' test statistics that render permutation tests asymptotically valid under a weak null
. Some scholars defend the constant-effects assumption more forthrightly, regarding it as a convenient approximation that is preferrable to the shortcomings of parametric methods, such as their sensitity to assumptions about tail behavior \citep{rosenbaum2010design} or inability to account for complex treatment assignments \citep{HoImai06a}.

It is also possible to model heterogeneity explicitly. \citet{Rosenbaum02a} demonstrates, for example, that permutation tests can be used to assess multiplicative, Tobit, quantile, and attributable effects, each of which permits certain precise kinds of heterogeneous additive effects \citep[see also][]{HoImai06a, Bowers13a}. 
Indeed, in principle there is no barrier to using RI to assess any arbitrary null hypothesis, so long as the hypothesis is sharp in the sense that it fully specifies the unit-level treatment effects (see, e.g., \cite{Li:2016tw}). 
In general, however, testing arbitrary sharp null hypotheses does not provide informative inferences because the parameter space is typically too unwieldy (with $n$ units, the space of possible effects is $n$-dimensional---assuming no spillover!).


While the defenses described above are reasonable, all presume that randomization tests are valid \emph{only} as tests of a sharp null hypothesis. 
As such, they do not directly address the concerns of critics who regard sharp nulls as inherently  ``restrictive,'' ``uninteresting,'' and ``academic.''
They also generally focus on testing rather than interval estimation.
We show, by contrast, that many randomization tests are also valid under a bounded null hypothesis, thus substantially expanding their applicability and permitting  assessment of the size as well as significance of treatment effects.

\section{Intuition for Validity under the Bounded Null}
\label{sec:illustration}

At the broadest level, Fisher and Neyman shared the same goal: making inferences about the effects of treatment on a given sample of units, based solely on the assumption that treatment was randomly assigned.
In \citeauthor{Neyman23a}'s (\citeyear{Neyman23a}) terms, both were interested in the differences in potential outcomes under treatment and control, $\tau_i = Y_i(1) - Y_i(0)$, for $n$ units indexed by $i$. Under the assumption that the potential outcomes $Y_i(1)$ and $Y_i(0)$ depend only on $i$'s own observed treatment status,\footnote{
  This is known as the stable unit treatment value assumption, or SUTVA \citep{rubin1980randomization}.
} these quantities of interest are fully defined by the \textit{potential-outcome schedule} $\sch$, which specifies all the potential outcomes in the sample.\footnote{
  This is what \citet[323--4]{Rubin05a} calls the ``science'' table. We use ``potential-outcome schedule'' to echo \citeauthor{Freedman09a}'s (\citeyear{Freedman09a}) term ``response schedule.''
} Drawing inferences about the unobserved elements of the potential-outcome schedule is the core task of causal inference.



As illustration, consider a sample of 16 units, 8 of which have been randomly assigned to treatment ($W_i=1$) and 8 to control ($W_i=0$). Table~\ref{tab:example_table} presents what we know about the sample. For each unit, only one potential outcome is observed; the other potential outcome is missing, and so is each unit's treatment effect. 
Suppose that we are interested in assessing the alternative hypothesis that units were positively affected by the treatment. Based on the observed outcomes, we calculate a treated--control difference of means of $t^{\mathrm{obs}} = \bar{Y}_1 - \bar{Y}_0 = +1.13$. How unlikely is a difference of means this large, relative to what would be expected by chance?

\begin{table}[ht]
\centering
\begin{tabular}{rrrrrr}
  \hline
 $i$ & $W_i$ & $Y_i$ & $Y_i(0)$ & $Y_i(1)$ & $\tau_i$ \\ 
  \hline
  1 &   0 & $-$0.90 & $-$0.90 & ? & ? \\ 
  2 &   0 & 0.18 & 0.18 & ? & ? \\ 
  3 &   0 & 1.59 & 1.59 & ? & ? \\ 
  4 &   0 & $-$1.13 & $-$1.13 & ? & ? \\ 
  5 &   0 & $-$0.08 & $-$0.08 & ? & ? \\ 
  6 &   0 & 0.13 & 0.13 & ? & ? \\ 
  7 &   0 & 0.71 & 0.71 & ? & ? \\ 
  8 &   0 & $-$0.24 & $-$0.24 & ? & ? \\ 
  9 &   1 & 2.98 & ? & 2.98 & ? \\ 
  10 &   1 & 0.86 & ? & 0.86 & ? \\ 
  11 &   1 & 1.42 & ? & 1.42 & ? \\ 
  12 &   1 & 1.98 & ? & 1.98 & ? \\ 
  13 &   1 & 0.61 & ? & 0.61 & ? \\ 
  14 &   1 & $-$0.04 & ? & $-$0.04 & ? \\ 
  15 &   1 & 2.78 & ? & 2.78 & ? \\ 
  16 &   1 & $-$1.31 & ? & $-$1.31 & ? \\ 
   \hline
\end{tabular}
\caption{
  The potential-outcome schedule $\sch$ for our 16-unit illustration.
} 
\label{tab:example_table}
\end{table}

Answering this question requires comparing $t^{\mathrm{obs}}$ to its reference distribution under some null hypothesis $\rH_0$. In the Fisherian paradigm, $\rH_0$ consists of an $n$-vector $\vtau^0$ of treatment effects, based on which we can create a null potential-outcome schedule $\schimp_{\vtau^0}$ with  the missing potential outcomes filled in. Using the imputed $\schimp_{\vtau^0}$, we can ``re-run'' our experiment on the same units and calculate the test statistics that would have been observed under alternative permutations of treatment assignment.
The collection of these values across permutations constitutes the statistic's reference distribution under the null hypothesis, conditional on the observed data.\footnote{
  We refer to alternative treatment assignments as ``permutations,'' but strictly speaking treatment assignment should only be permuted if the size of the treatment group is fixed by design. If, instead, treatment is assigned by another  mechanism, such as binomial randomization, then the reference distribution can be generated accordingly, or, alternatively, one could condition on the total of units treated. See, e.g., \citet{Hennessy:kl} for discussion. Regardless, our arguments here readily extend to alternative randomization schemes.}
The proportion of permutations with a value of the test statistic at least as large as $t^{\mathrm{obs}}$ is the $p$-value under $\rH_0$.

This procedure is most straightforward for the sharp null of no effect. Under this hypothesis, $\tau_i = 0\; \forall i$, so the missing potential outcomes imputed under this hypothesis simply equal the observed outcomes (see Table \ref{tab:NullScience}, columns 4--5). The same procedure, however,  may be used for any arbitrary vector $\vtau^0$ of hypothesized treatment effects. If $W_i=1$, we simply impute the missing $Y_i(0)$ as  $\tilde{Y}_i(0)= Y_i - \tau^0_{i}$. Likewise, if $W_i=0$ we impute $\tilde{Y}_i(1)  = Y_i + \tau^0_{i}$. Columns 7--9 of Table \ref{tab:NullScience} illustrate this procedure for a constant-effect null of $\tau_{i} = -1 \; \forall i$, and columns 10--12 do so for a bounded 
 null under which most effects are 0 but two are negative. For any such sharp null, we can generate the reference distribution by repeatedly permuting the treatment variable $W$, determining potential outcomes that would have been observed under that treatment assignment, and calculating the value of the test statistic in each permutation.\footnote{
  See, for example, \citet[660]{Ding16a}. It is worth noting that this procedure differs slightly from that described by \citet[e.g.,][44]{Rosenbaum02a}, who instead proposes testing whether the imputed potential outcomes under control, $\tilde{Y}_i(0) = Y_i - W_i\tau^0_{i}$, satisfy the sharp null of no effects. A disadvantage of the Rosenbaum procedure is its arbitrary choice of $\tilde{Y}_i(0)$ rather than $\tilde{Y}_i(1)$ as a baseline. This choice can affect the results of the test if, for example, the difference of means is the test statistic and the null stipulates heterogeneous additive effects. For intuition on this point, observe that the vector $\mathbf{\tilde{Y}}(1)$ imputed under the non-superiority null (Table \ref{tab:NullScience}, column 11), because it incorporates the $-2$ treatment effect for unit 2, deviates more strongly from the sharp null than does $\mathbf{\tilde{Y}}(0)$, which incorporates only the smaller $-1$ effect for unit 12. The exact $p$-values under the no-effects null are 0.025 for $\mathbf{\tilde{Y}}(1)$ and 0.035 for $\mathbf{\tilde{Y}}(0)$. In expectation, either test is valid, but for any realized treatment assignment their results can differ.
}

\begin{table}[ht]
  \centering
  \begin{tabular}{|rrr|rrr|rrr|rrr|}
    \hline
    \multicolumn{3}{|c|}{Observed Data} 
    & \multicolumn{3}{c|}{$\rH_0$: No Effects} 
    & \multicolumn{3}{c|}{$\rH_0$: Constant Effect}
    & \multicolumn{3}{c|}{$\rH_0$: Non-Superiority}\\
    \hline
     $i$ & $W_i$ & $Y_i$ & $\tilde{Y}_i(0)$ & $\tilde{Y}_i(1)$ & $\tau^0_{i}$
    & $\tilde{Y}_i(0)$ & $\tilde{Y}_i(1)$ & $\tau^0_{i}$ & $\tilde{Y}_i(0)$ & $\tilde{Y}_i(1)$ & $\tau^0_{i}$ \\ 
    \hline
1 &   0 & $-$0.90 & $-$0.90 & $-$0.90 &   0 & $-$0.90 & $-$1.90 &  $-$1 & $-$0.90 & $-$0.90 & 0 \\ 
  2 &   0 & 0.18 & 0.18 & 0.18 &   0 & 0.18 & $-$0.82 &  $-$1 & 0.18 & $-$1.82 & $-$2 \\ 
  3 &   0 & 1.59 & 1.59 & 1.59 &   0 & 1.59 & 0.59 &  $-$1 & 1.59 & 1.59 & 0 \\ 
  4 &   0 & $-$1.13 & $-$1.13 & $-$1.13 &   0 & $-$1.13 & $-$2.13 &  $-$1 & $-$1.13 & $-$1.13 & 0 \\ 
  5 &   0 & $-$0.08 & $-$0.08 & $-$0.08 &   0 & $-$0.08 & $-$1.08 &  $-$1 & $-$0.08 & $-$0.08 & 0 \\ 
  6 &   0 & 0.13 & 0.13 & 0.13 &   0 & 0.13 & $-$0.87 &  $-$1 & 0.13 & 0.13 & 0 \\ 
  7 &   0 & 0.71 & 0.71 & 0.71 &   0 & 0.71 & $-$0.29 &  $-$1 & 0.71 & 0.71 & 0 \\ 
  8 &   0 & $-$0.24 & $-$0.24 & $-$0.24 &   0 & $-$0.24 & $-$1.24 &  $-$1 & $-$0.24 & $-$0.24 & 0 \\ 
  9 &   1 & 2.98 & 2.98 & 2.98 &   0 & 3.98 & 2.98 &  $-$1 & 2.98 & 2.98 & 0 \\ 
  10 &   1 & 0.86 & 0.86 & 0.86 &   0 & 1.86 & 0.86 &  $-$1 & 0.86 & 0.86 & 0 \\ 
  11 &   1 & 1.42 & 1.42 & 1.42 &   0 & 2.42 & 1.42 &  $-$1 & 1.42 & 1.42 & 0 \\ 
  12 &   1 & 1.98 & 1.98 & 1.98 &   0 & 2.98 & 1.98 &  $-$1 & 2.98 & 1.98 & $-$1 \\ 
  13 &   1 & 0.61 & 0.61 & 0.61 &   0 & 1.61 & 0.61 &  $-$1 & 0.61 & 0.61 & 0 \\ 
  14 &   1 & $-$0.04 & $-$0.04 & $-$0.04 &   0 & 0.96 & $-$0.04 &  $-$1 & $-$0.04 & $-$0.04 & 0 \\ 
  15 &   1 & 2.78 & 2.78 & 2.78 &   0 & 3.78 & 2.78 &  $-$1 & 2.78 & 2.78 & 0 \\ 
  16 &   1 & $-$1.31 & $-$1.31 & $-$1.31 &   0 & $-$0.31 & $-$1.31 &  $-$1 & $-$1.31 & $-$1.31 & 0 \\ 
    \hline
    \multicolumn{3}{|c|}{$t^{\mathrm{obs}} = +1.13$} 
    & \multicolumn{3}{c|}{$p=0.040$} 
    & \multicolumn{3}{c|}{$p=0.002$}
    & \multicolumn{3}{c|}{$p=0.027$}\\
    \hline
  \end{tabular}
  \caption{
    Potential-outcome schedules imputed under the sharp null hypotheses of no effects (columns 4--6), a constant effect of $-1$ (columns 7--9), and non-superiority (columns 10--12). The bottom row lists the observed difference of means ($t^{\mathrm{obs}}$) and its one-sided permutation $p$-values under the three null hypotheses.
  }
  \label{tab:NullScience}
\end{table}

The bottom row of Table \ref{tab:NullScience} first reports $t^{\mathrm{obs}}$, the observed difference of means, and then lists the $p$-values of this statistic under each sharp null hypothesis. Notice that the $p$-values under the constant-effect null and the bounded null are both smaller than the $p$-value under the null of no effects whatsoever. This is no coincidence. Rather, as we later prove, the $p$-value under \textit{any} sharp 
 null that satisfies $\tau_{i} \leq 0\ \forall i$ is guaranteed to be no larger than the $p$-value under the sharp null of no effects ($\tau_{i} = 0\ \forall i$).\footnote{
   This presumes that the test statistic is larger in the alternative than under the sharp null. If it is smaller in the alternative, then the statement holds for any null that satisfies $\tau_{i} \geq 0\ \forall i$.
} This result follows from the fact that the reference distribution generated under the null of no effects weakly stochastically dominates the distribution under any non-superiority null. This fact is illustrated visually in Figure~\ref{fig:perm_demo_plot}, which plots the reference distribution for the no-effects null (solid black line), the constant-effects null of $-1$, and ten randomly generated nulls with heterogeneous treatment effects between $-1$ and 0. 
Observe that the density lines of the heterogeneous-effect nulls are all to the left of the solid no-effects density line and to the right of the dashed constant-effect line. Consequently, the cumulative density greater than $t^{\mathrm{obs}}$ (vertical dotted line)---that is, the $p$-value---is largest for the no-effect null, smallest for the constant-effect null, and somewhere in the middle for each of the heterogeneous nulls.

\begin{figure}
\center
\includegraphics[width=0.7\textwidth]{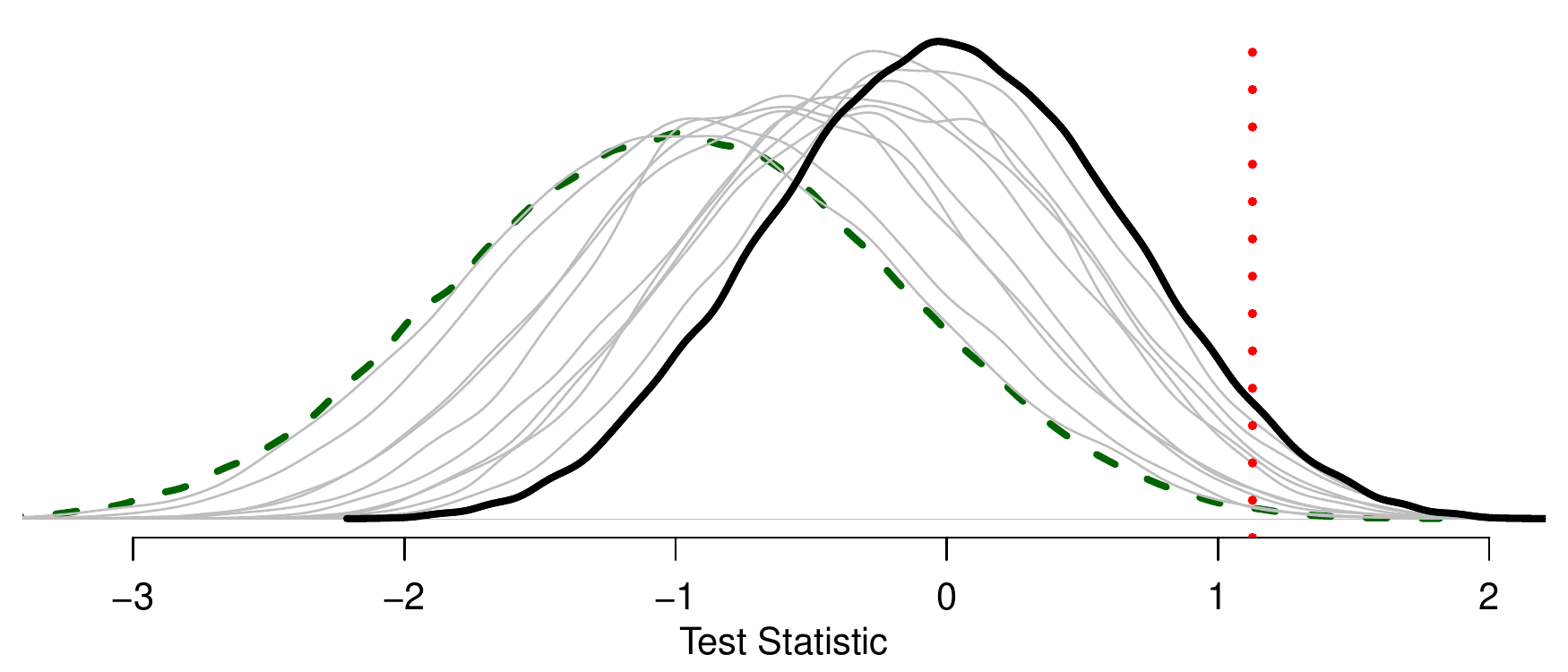}
\caption{
  Reference distributions under different null hypotheses.  The dark black line is the reference distribution of $t = \bar{Y}_1 - \bar{Y}_0$ under the sharp null hypothesis of no effects.  The red dotted line denotes the observed $t^{\mathrm{obs}}$.  The $p$-value is the area under the curve to the right of this, and as it is small we would reject the null. The dashed green line is the reference distribution under the null hypothesis of a constant treatment effect of $\tau^0 = -1$.  The grey lines are a sample of 10 possible reference distributions for ten different non-superiority nulls of no positive effects.  They are all stochastically lower than the sharp zero null, and thus have lower $p$-values.
}
\label{fig:perm_demo_plot}	
\end{figure}

This result has several implications. First, it means that if the sharp null hypothesis of no effect $\tau_{i} = 0\ \forall i$ can be rejected at level $\alpha$, so can any null hypothesis such that $\tau_{i} \le 0\ \forall i$.\footnote{
  Again, assuming without loss of generality that the test statistic is large under the alternative.
} In other words, tests of the no-effect null are conservative tests of the more general bounded null (see Figure \ref{fig:NonSupNull} for a visual representation of the relationship between Fisher's no-effect null and the bounded null). Moreover, this result extends to any constant-effect null, in that rejection of $\rH_0: \tau_{i} = \tau^0\ \forall i$ implies rejection of $\rH_0: \tau_{i} \le \tau^0\ \forall i$. This observation is particularly important for confidence intervals (CIs), which in RI are defined as the collection of sharp null hypotheses not rejected at a given significance level. Typically, randomization CIs are calculated under a constant-effect assumption, but the above result suggests an alternative interpretation that does not require this assumption. Under this re-interpretation, a one-sided randomization CI for a constant effect is also a valid CI for the lower bound on the maximal unit-level effect, $\tau^{\mathrm{max}}$. Thus, an $\alpha$-level CI of $[L, \infty)$ will miss the true $\tau^{\mathrm{max}}$ with probability $\alpha$, and we can conclude with $100 \times (1-\alpha)$\% confidence that at least some units had a treatment effect as large as $L$.

\begin{figure}
  \centering
  \includegraphics[width=0.6\textwidth]{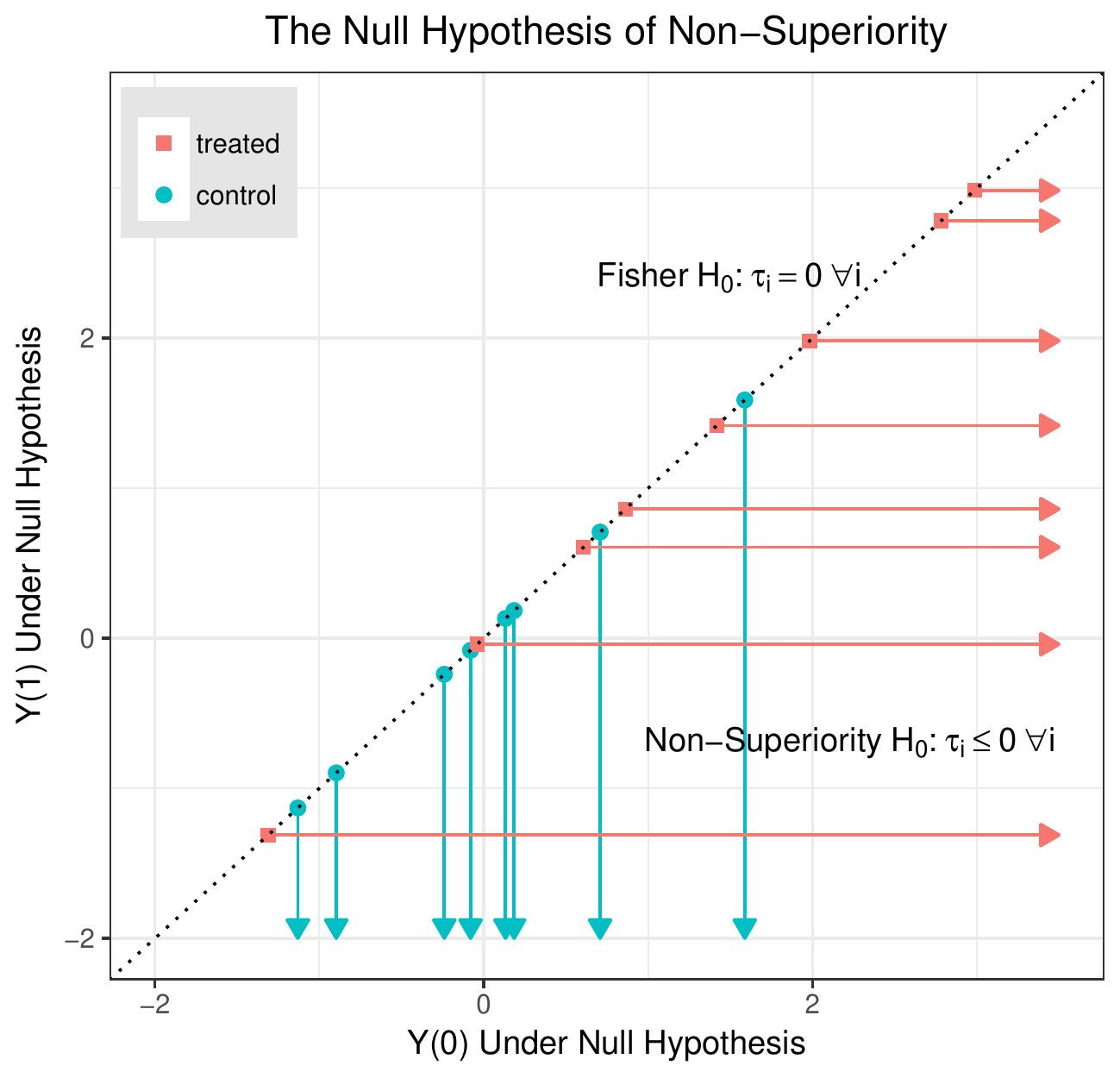}
  \caption{
    The null hypothesis of non-superiority. The horizontal and vertical axis indicate, respectively, $\tilde{Y}_i(1)$ and $\tilde{Y}_i(0)$: the treated and control potential outcomes imputed under the null hypothesis. The dotted 45$^{\circ}$ line represents Fisher's sharp null of no effects, under which $\tilde{Y}_i(1)=\tilde{Y}_i(0)$. The red squares indicate the observed treated outcomes ${Y}_i(1)$, and the horizontal red lines indicate their possible values of $\tilde{Y}_i(0)$ under the bounded 
     null. The blue circles and vertical blue lines indicate the ${Y}_i(0)$ and possible $\tilde{Y}_i(1)$ for units actually assigned to control. 
}
  \label{fig:NonSupNull}
\end{figure}

For intuition, consider the example of a newly developed drug, whose side effects on pain we wish to compare to those of an existing drug. In particular, we wish to assess whether the new drug increases any subject's pain level $Y_i$ relative to the existing drug. Given random assignment to the existing drug ($W_i = 0$) and the new ($W_i = 1$), we can do so using an appropriate one-sided permutation test of the null hypothesis $\rH_0: Y_i(1) - Y_i(0) \equiv \tau_{i} = \tau^0\ \forall i$ for a sequence of $\tau^0$ values. If $L$, the smallest value of $\tau^0$ that cannot be rejected at $\alpha = 0.1$, is greater than 0, then we can conclude with 90\% confidence that the new drug caused at least one subject at least $L$ more pain than the existing drug would have. In other words, we can conclude not only that the sharp null of no effects is implausible, but that at least one subject was adversely affected by the new drug.


\section{Formal Exposition of Validity under the Bounded Null}
\label{sec:perm-tests-non}

Having illustrated the intuition behind our argument, we now turn to formal exposition. As noted above, permutation tests are conventionally conducted under a sharp null hypothesis that precisely specifies the $n$-vector of unit-level treatment effects $\vtau$. Let $\sharpH$ denote the sharp null hypothesis $\vtau = \vtau^0$, where $\vtau^0$ is a vector of hypothesized treatment effects, not necessarily equal to 0 or any other constant. Testing such a sharp null hypothesis entails first choosing a test statistic $T(W, Y)$, which is a function of the treatment assignment vector $W$ and the observed outcomes $Y$.
\footnote{We are simplifying the exposition by considering univariate outcomes and by ignoring covariates. More generally, however, a test statistic may be a function of multiple outcome variables as well as of covariates.}
Since $Y$ is itself a function of $W$ and the potential outcomes $Y(0)$ and $Y(1)$, with $Y_i = Y_i(W_i)$, we can also write the test statistic as $T(W, Y(0), Y(1)) = T(W, \sch)$. $\sch$ being fixed, the randomness in $T(W, \sch)$ comes only from the randomness in $W$.\footnote{Strictly speaking, the $T(W, \sch)$ notation allows for statistics not directly observable (e.g., ones that use both potential outcomes of the units). We use this notation for later clarity in the formal argument, with the implicit restriction of $T(W, \sch)$ that could also be represented as $T(W, Y)$, a function of the observed data.} 
Examples of test statistics include the treated-control difference of means, 
\begin{align}	
  T(W, \sch) = \sum W_i Y_i(1) / \sum W_i - \sum (1-W_i) Y_i(0) /\sum (1-W_i)
  \label{eq:mean_difference},
\end{align}
but there are many other options, several of which we discuss later in the paper.

The observed value of the test statistic is $\tobs = T(\vw^{\text{obs}}, \vy^{\text{obs}})$, where $\vw^{\text{obs}}$ and $\vy^{\text{obs}}$ are, respectively, the observed treatment and outcome vectors.  To evaluate whether a test statistic value as large as $\tobs$ would be unusual under the null hypothesis, we compare it to its permutation distribution under the null. Being sharp, $\sharpH$ allows us to impute the potential-outcomes schedule under the null, $\schimp_{\vtau^0} = \schimp( \wobsV, \yobsV, \sharpH )$, through the relations
 \[
\tilde{Y}_{i}(0) =
\begin{cases}
  w_i = 0 & y^{\text{obs}}_i\\
  w_i = 1 & y^{\text{obs}}_i - \tau^0_{i}
\end{cases}
\]
and 
\[
\tilde{Y}_{i}(1) =
\begin{cases}
  w_i = 0 & y^{\text{obs}}_i + \tau^0_{i}\\
  w_i = 1 & y^{\text{obs}}_i.
\end{cases}
\] 
Using $\schimp_{\vtau^0}$, we can then compute the test statistic value that would have been observed under any alternative realization of $W$. Let $\vw^*$ denote a random draw from the space of potential treatment assignments, and let $t^* = T(\vw^*, \schimp_{\vtau^0})$ be the value of the test statistic given $\vw^*$ and $\schimp_{\vtau^0}$.\footnote{
   The assignment vector $W$ is random according to an assignment mechanism which returns a specific $\vw$ with given probability $\pr{ W = \vw }$. For instance, in the completely randomized design, $\pr{ W = \vw} = {N \choose N_T}^{-1}$ for any $\vw$ such that $\sum w_i = N_T$ for some pre-specified $N_T$. In most typical experiments, all treatment assignments that have non-zero probability are equiprobable.
} The $p$-value of $\tobs$ under $\sharpH$ is thus the probability across permutations of observing a test statistic value at least as large as $\tobs$:
\begin{equation}
  \label{eq:2}
  p_{\vtau^0} \equiv \pr{ T( {\vw}^*, \schimp_{\vtau^0} ) \geq \tobs}.
\end{equation}
This exact $p$-value can be estimated to arbitrary precision by sampling $J$ random permutations of treatment $\vw^*_j$ and calculating
\begin{equation}
  \label{eq:3}
  \hat{p}_{\vtau^0} = 
  \frac{1}{J} \sum_{j=1}^J \ind{ T( \vw^*_j, \schimp_{\vtau^0} ) \geq \tobs } \approx p_{\vtau^0}.
\end{equation}

We next show that, for a class of \textit{effect-increasing} test statistics, a test of the sharp null is actually a valid test of the much more general \emph{bounded null hypothesis}, under which the treatment effects are bounded on one side by the sharp null but otherwise may be arbitrarily heterogeneous. When the test statistic is an increasing function of the treatment effects, such as the difference of means, this bounded null is one of \textit{non-superiority}---that is, one bounded above by the sharp null:
\begin{align*}
\nonsupH : \tau_i \leq \taumax_{i} \equiv \tau^0_{i} \quad \forall i \in 1\ldots n . \tag{Non-Superiority}
\end{align*}
Analogously, when the test statistic is decreasing in treatment effects, the weak null is a lower bound  of \textit{non-inferiority}:
\begin{align*}
\rH_{\vtaumin} : \tau_i \geq \taumin_{i} \equiv \tau^0_{i} \quad \forall i \in 1\ldots n. \tag{Non-Inferiority}
\end{align*}
For ease of exposition we here focus on the non-superiority null $\nonsupH$, which we simply refer to as the bounded null, but all our results can be extended to $\noninfH$ by multiplying $Y$ by $-1$.

Unlike the sharp null $\sharpH$, which corresponds to a single null potential-outcomes schedule $\schimp_{\vtau^0}$, the bounded null $\nonsupH$ corresponds to an infinitely large set of such schedules that are consistent with the observed data and the treatment-effect bound (e.g., every set of points on the lines in Figure \ref{fig:NonSupNull}). Thus, if we reject this null, we are rejecting a set of null schedules (and associated treatment-effect vectors) rather than a single point null. Because each weak null permits infinitely many possible potential-outcome schedules, each of which implies a different null distribution for the test statistic, no single $p$-value will be exact (i.e., have a false-positive rate of exactly $\alpha$) for all possible schedules. We can show, however, that for a class of test statistics the $p$-value associated with any of these null distributions will be bounded above by the $p$-value under $\sharpH$, making the sharp null $p$-value valid but possibly conservative for the bounded null. (A hypothesis test is \emph{conservative} if, for any nominal significance level $\alpha$, the true probability of incorrectly rejecting the null hypothesis is no greater than $\alpha$.)



The above property holds for permutation tests that employ an effect-increasing test statistic. To define this class of statistics, we must first introduce the notion of ordering potential-outcome schedules:
\paragraph{Definition: Ordering of Potential-Outcome Schedules.}
Two potential-outcome schedules $\sch$ and $\sch'$ are ordered as $\sch \preceq \sch'$ if $Y_i(1) \leq Y_i'(1)$ and $Y_i(0) \geq Y_i'(0)\ \forall i \in 1\dots n$. That is, $\sch \preceq \sch'$ if and only if no unit's potential outcome under treatment is smaller in $\sch'$ than in $\sch$ and no unit's potential outcome under control is larger in $\sch'$ than in $\sch$. An immediate consequence of such ordering is that the individual treatment effects are also ordered: $\tau_i \leq \tau_i'\ \forall i$. 

Our class of statistics is then defined as those that satisfy the following:
\paragraph{Definition: Effect-Increasing (EI).}  
A test statistic $T$ is \emph{effect-increasing} if, for two potential-outcome schedules $\sch$ and $\sch'$, $\sch \preceq \sch'$ 
 implies $T(\vw, \sch ) \leq T(\vw, \sch' )$ for all allowed realizations $\vw$ of the  treatment variable $W$. In other words, a test statistic $T$ is effect-increasing if it is weakly increasing in the potential outcomes under treatment and weakly decreasing in the potential outcomes under control \citep[cf.\@][37--8]{Rosenbaum02a}. Since $\sch \preceq \sch'$ implies $\tau_i \leq \tau_i'\ \forall i$, an EI statistic is also increasing in the individual treatment effects (hence the label ``effect-increasing''). 

For effect-increasing statistics we have the following proposition:
\paragraph{Proposition 1.} If $\vtau_0 = \vtaumax$, a permutation test of $\sharpH$ is a conservative (and thus valid) test of $\boundedH$ in that for all $\sch_h \in \boundedH$
\[ \pr{ \mbox{Reject } \boundedH | \sch_h } = \pr{ T(\vw^*, \schimp_{\vtau^0})
 \geq T( \vw, \sch_h ) } \leq \alpha , \]
with the probabilities taken across both $\vw$ and $\vw^*$, each random draws from the assignment mechanism. (For a formal proof of this propoosition, see Appendix \ref{sec:proof-validity-under}.)

This result means that if we consider a test of a given sharp null as a test of the associated bounded null, we still have a valid test. In particular, rejecting $\boundedH$ when the nominal $p$-value for a permutation test of $\sharpH$ is less than $\alpha$ is a valid testing procedure for $\boundedH$.

\section{Confidence Intervals for Maximum/Minimum Effects}
\label{sec:confidence-intervals}

In most settings, researchers are not content merely to test a single null hypothesis. Rather, they wish to summarize the information the data convey regarding quantities of interest. One way to do this is to estimate confidence intervals that indicate the range of plausible parameter values. Because of the duality between hypothesis testing and interval estimation, a general method of CI construction is to define the $1-\alpha$\% CI as the set of parameter values that cannot be rejected at some significance level $\alpha$. In the case of a one-sided $t$ test, for example, the lower bound of the 90\% CI for the mean $\mu$ is the smallest value of $\mu^0$ such that the weak null hypothesis $\mu \le \mu^0$ cannot be rejected at $\alpha=0.1$. Analogously, in the case of an effect-increasing randomization test, the lower bound of the CI for the maximum unit-level effect $\taumax_{i}$ is the smallest value of $\tau^0$ such that the non-superiority hypothesis $\tau_i \leq \tau^0\; \forall i$ cannot be rejected.




For sharp null inference, confidence intervals are generally obtained by inverting a sequence of sharp-null level-$\alpha$ tests of hypotheses $\sharpH$.
Typical practice would be to consider a sequence of constant shift hypotheses $\rH_{\tau^0}: \tau_i = \tau^0\ \forall i$.
For each candidate value $\tau^0_h = \tau^0$ we calculate \[
p( \tau^{0}_h) = \pr{ T( \vw^*, 
\schimp_{\tau^{0}_h}) \geq \tobs },
\] and if $p( \tau^0_h ) \leq \alpha$, we conclude that $\tau^0_h$ is implausible.
This gives a confidence set of plausible $\tau^0$ values of
\[ \mathrm{CI} \equiv \left\{ \tau^0_h : p( \tau^0_h ) \geq \alpha \right\}. \]
The confidence sets are random depending on the randomization.
They are valid in the sense that if the treatment effect is in fact constant ($\sch \in \rH_{\tau}$ for some $\tau$), then the confidence set CI will contain $\tau$ with probability at least $1-\alpha$.
Unfortunately, if the treatment effect is not constant, there is no immediate reason for CI to contain any particular summary of the treatment effects (e.g., the ATE). This is one of the primary complaints against permutation inference. 

By viewing these as tests of a bounded null, however, the associated confidence interval does in fact have a general interpretation that does not depend on the implausible assumption of constant effects. 
In particular, the conventional randomization CI for a constant effect can  be interpreted as a confidence set on the maximum treatment effect in the sample (or minimum, in the case of a non-inferiority hypothesis). In other words, given a one-sided CI $[L,\infty)$ we can say that we are at least $1-\alpha$ confident that some units have a treatment effect of at least $L$. (For a proof of this proposition, see Section \ref{sec:proof-validity-ci} in the Appendix.) This statement does not depend on any specific structure on the individual effects; they may be arbitrarily heterogeneous. That being said, the more homogenous the effects, the more individual effects we would expect to be in the CI. Of course, in the limiting case of no effect heterogeneity (i.e., a constant effect), $[L,\infty)$ will, as discussed above, contain all the individual effects with probability $1-\alpha$.




It should be emphasized that the permutation CI for the maximum effect will have correct coverage regardless of the test statistic used, as long as that statistic is effect-increasing. The CI itself will, however, vary depending on the test statistic's power against different alternative hypotheses. In particular, unless the treatment effects are close to constant, using a statistic sensitive to the central tendency may result in relatively uninformative confidence bounds for the maximum. Thus, if heterogenous effects are expected, it may be preferable to use a statistic that is sensitive to the largest effects, such as the Stephenson rank sum (discussed below).

\section{Effect-Increasing Test Statistics}
\label{sec:test-stat-power}

All effect-increasing test statistics are valid under the bounded null, but what are they? We next show that the difference of means, the Wilcoxon rank sum, and other common test statistics are all effect-increasing. 
We also note that others, in particular the studentized difference of means, are not. 
We then briefly touch on the issue of statistical power against different alternative hypotheses.

As discussed above, an effect-increasing statistic grows as the treatment potential outcomes get larger or the control potential outcomes get smaller.
The difference in means, $T( Y, W ) = \bar{Y}_1 - \bar{Y}_0$, is an intuitive example of such a statistic: if we increase any $Y_i(1)$ then the difference in means will either go up (if $W_i = 1$) or remain unchanged (if $W_i = 0$).
Similarly, if we lower any $Y_i(0)$ the control side mean will either shrink, which will raise $t$, or stay the same.
This general argument applies for several broad classes of statistics, as we present below in two lemmas.  
Proofs are in Appendix~\ref{sec:proof-that-rank}.

First, statistics that look like differences in means, or sums of rescaled treatment outcomes, are all effect-increasing:
\begin{lemma}
 Let $Q(\cdot)$ be a non-decreasing function of outcomes and $a(W), b(W) \geq 0$ be two scaling numbers that can depend on treatment assignment.  
 Then the class of statistics defined as 
\[ T( Y, W ) = a(W) \sum_{i=1}^n W_i Q( Y_i ) - b(W) \sum_{i=1}^n (1-W_i) Q( Y_i ) \]
are effect-increasing.
The difference-in-means is a special case with $a(W) = 1/n_T$, $b(W) = 1/n_C$, and $Q(Y_i) = Y_i$.\footnote{As a mild technical condition, define $a(W)$ and $b(W)$ as constant when there are no treated or no control units in the case of assignment mechanisms that have variable numbers of units treated.}
\end{lemma}

Furthermore, many common rank statistics are also effect-increasing:
\begin{lemma}
For continuous outcomes any statistic of the form
\[ T(W, Y) = a(W) \sum_i W_i Q( R_i ) - b(W) \sum_i (1 - W_i) Q(R_i), \]
where $R_i$ is the rank of $Y_i$, $a(W)$ and $b(W)$ are nonnegative functions of the assignment vector, and $Q(R_i)$ is some non-decreasing function of the ranks, is effect-increasing.
\end{lemma}
For many statistics this result extends to outcomes with ties.
In particular, see Appendix~\ref{sec:proof-that-rank} for derivations for the Wilcoxon rank sum statistic and Stephenson rank sum statistics (see below).

The EI property does \textit{not} hold, however, if the difference of means is ``studentized'' by a consistent estimate of its standard error: \[
t = \frac{\bar{Y}_T - \bar{Y}_C}{\sqrt{s^2_T / N_T + s^2_C / N_C}}.
\] 
The $t$ statistic is not EI because a large increase in one unit's treated outcome can have such a large effect on the standard deviation $s_T$ that it outweighs the effect on the mean $\bar{Y}_T$, thus decreasing the statistic overall. 
More obviously, statistics not sensitive to one-sided location shifts---such as the absolute difference of means, the difference of variances, and the two-sided Kolmogorov-Smirnov statistic---are also not effect-increasing.

Any effect-increasing statistic will give a valid test of the general bounded null.
They are not all the same, however.
In particular, a given statistic can be more or less sensitive to different alternatives, which can influence power.
One class of statistics particularly well-suited to this circumstance are Stephenson rank statistics \citep{Stephenson81a, Stephenson85a}.
Stephenson rank statistics are two-sample statistics with a score function of \[ 
Q(R_i) =
\begin{cases}
  {R_i - 1 \choose s - 1} & R_i \geq s \\
  0 & \mbox{ otherwise} 
\end{cases}
\] for some fixed integer $s \ge 2$. 
Stephenson rank statistics are equivalent to summing the number of subsets of size $s$ in which the largest response is in the treated group. The Stephenson rank statistic with $s=2$ is almost identical to the Wilcoxon rank sum. However, as $s$ increases beyond 2, the Stephenson ranks place more and more weight on the largest responses.

Stephenson rank statistics are particularly interesting in the context of this paper due to their power to detect uncommon-but-dramatic responses to treatment \citep{Rosenbaum07a}. Intuitively, this is because as the subset size $s$ increases, it becomes increasingly likely that the largest response in a given subset will be one with an unusually large treatment effect.\footnote{
  Examining the asymptotic relative efficiency of a closely related class of test statistics, \citet[196]{Conover88a} find that when a only small fraction of treated respond, the optimal subset size $s$ is between 5 and 6.
} Thus, compared to the difference of means and the Wilcoxon rank sum, whose power is greatest against a constant location shift, the Stephenson ranks have greater power against alternatives under which effects are heterogeneous and a few are highly positive (relative to the null). 
This sensitivity to extreme treatment effects leads to tighter confidence intervals for the maximum effect when the maximum differs greatly from the mean or median. It is even possible for a Stephenson rank test to reject the non-superiority null for positive values when the ATE estimate is \textit{negative}, if some treatment effects are sufficiently positive. 
Thus, when treatment effects are heterogeneous, the behavior of the Stephenson test can differ markedly from the rank sum or difference of means, while, like them, still providing a valid test of the bounded null hypothesis.


\section{Practical and Theoretical Relevance}
\label{bounded}




So far, we have demonstrated that effect-increasing randomization tests can be interpreted as tests of a bounded null hypothesis,
 that inverting a sequence of such tests produces valid confidence intervals for the maximum effect, and that many familiar test statistics are effect-increasing.
But do these formal results have any practical relevance, or are they as ``uninteresting and academic'' as the sharp null? We believe that both testing bounded null hypotheses and interval estimation for maximum or minimum effects are indeed relevant to many social-scientific questions. 

There are many settings where bounded null hypotheses are either invoked as technical assumptions or given special normative status; in either case, testing and possibly rejecting the bounded null is both scientifically interesting and practically meaningful. Many statistical techniques, for instance, rely on a monotonicity assumption that one causal variable has non-negative or non-positive effects on another. Instrumental-variable (IV) estimation, Manski bounds, mediation analysis, the signing of paths on directed acyclic graphs, and the signing of selection bias are all examples of such methods. Obviously, the monotonicity assumption is equivalent to the hypothesis that the unit-level effects are bounded at 0 on one side. In many social contexts, the monotonicity assumption is considered ``reasonable'' \citep[206]{MorganWinship07a}, ``plausible'' \citep[452]{Imbens15a}, or at least ``defensible'' \citep[217]{Gelman06a}, suggesting that rejecting it and other bounded null hypotheses can be substantively interesting. Even applied studies that explicitly invoke monotonicity, however, usually do not test it empirically, perhaps because doing so fits poorly with conventional methods' focus on average effects. As we show in Section \ref{sec:test-monot-an}, our re-interpretation of randomization inference provides a straightforward basis for evaluating monotonicity.

In addition to being invoked as substantive assumptions to justify using certain statistical methods, bounded null hypotheses frequently have special normative or theoretical significance. From a normative perspective, there are many circumstances where ``nonmaleficence,'' or avoidance of harm, is an important moral standard: even if the overall or expected effect of an action would be beneficial, one nevertheless has an obligation to avoid harming any individual \citep[e.g.,][chapter 5]{BeauchampChildress12a}. According to this standard, a policy or other intervention whose average or median effect is positive may nevertheless be ethically objectionable if it has a single negative effect. This emphasis on avoiding harm even if the net benefit is positive appears in many contexts, ranging from U.S. tort law to experimental moral philosophy \citep[e.g., in variants of the ``trolley problem'';][]{Thomson85a}.
A related ethical standard, prominent in economics, is ``Pareto improvement''---a change in the allocation of resources that increases average welfare while hurting no one.  If, for example, campaign contribution limits benefit all citizens, even campaign donors themselves \citep{Coate04a}, then from a consequentialist perspective there should be no objection to such limits. By contrast, a policy that benefits most but harms some, such as a free-trade agreement without an adequate compensation scheme for losers, may fail to pass ethical muster.\footnote{
  When compensation schemes are workable then finding a Pareto-optimal policy simplifies to maximizing the ATE and implementing such a compensation scheme. However, compensations schemes are often not workable. The outcome, such as quality-adjusted life years, may not be fungible. Compensation schemes may not be incentive compatible, or may be costly to implement. Most fundamentally, compensation schemes require estimating with sufficient precision every person's counterfactual gain or loss---that is, their individual treatment effect---which is generally impossible without strong assumptions.
  }
%
Whether couched as nonmaleficence or Pareto improvement, the ethical standard is again a bounded hypothesis, which we may wish to test empirically.

\begin{figure}
  \centering
  \includegraphics[width=\textwidth]{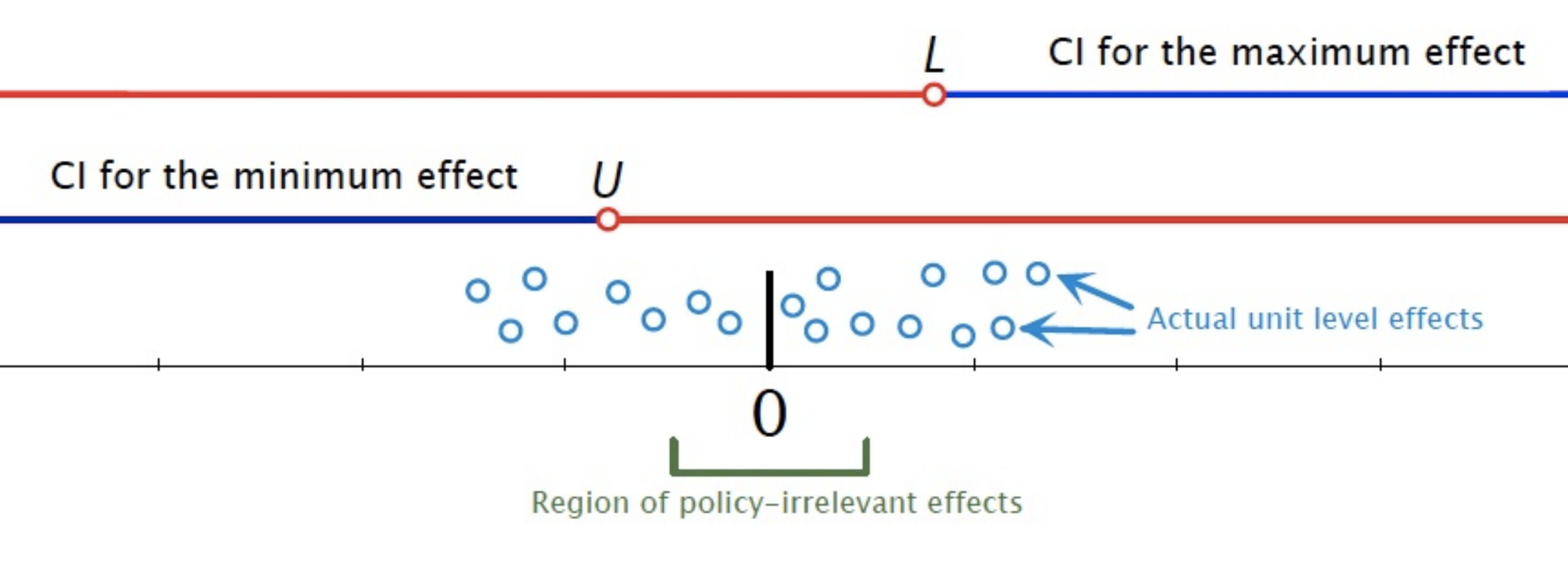}
  \caption{
    Two one-sided confidence intervals for maximum and minimum effects (blue lines). The confidence intervals $[L, \infty)$ and $(-\infty, U]$ do not overlap in this illustration, which is possible if they are based on a test statistic sensitive to the tails of the distribution of treatment effects (e.g., the Stephenson rank sum).
}
  \label{fig:2ci}
\end{figure}

Of course, scholars generally seek to infer not only the sign of treatment effects, but also their magnitude. Our results are relevant for these purposes as well. In some circumstances, for example, it may be of interest whether a treatment had any effect large enough to be ``policy relevant'' or otherwise substantively important. If such a relevance threshold can be stipulated \textit{a priori}, then an effect-increasing randomization test of the null hypothesis that all effects are bounded by this threshold provides an answer to this question. If the hypothesis is rejected, then we can conclude that there were in fact some policy-relevant effects. Randomization CIs provide a more general answer to this question, enabling readers with possibly differing thresholds to come to their own conclusions. With test statistics sensitive to the extremes, it is even possible to test and form confidence intervals for policy-relevant effects in opposite directions. If effects are sufficiently heterogeneous, the two confidence intervals may not even overlap, indicating the presence of both positive and negative effects beyond thresholds of substantive importance (see Figure \ref{fig:2ci}). This inference may take on particular significance in small samples. If we can conclude, for example, that at least one unit in a 20-unit sample had a policy-relevant effect, then, informally speaking, this provides evidence that 5\% of the larger population would experience such an effect. In short, both tests of bounded null hypotheses and confidence intervals for maximum or minimum effects can help answer questions of practical and theoretical relevance to applied researchers, a fact we illustrate with the applications that follow.

\section{Applications}
\label{sec:application}

We illustrate the practical implications of our results with re-analyses of two published studies. In the first, we show how randomization tests' validity under the bounded null provide a basis for assessing the assumption that an instrumental variable has monotonic effects on the treatment. In the second application, we show how RI can be used to make inferences about effect sizes in a 16-unit randomized experiment.

\subsection{Testing Monotonicity of an Instrumental-Variable Analysis}
\label{sec:test-monot-an}

The assumption that the instrument has monotonic effects on the treatment, though conventionally invoked for identification of IV estimates \citep{AngristImbensRubin96a}, is rarely evaluated in empirical applications. Recently, however, the issue of non-monotonicity has received attention in the active literature on school-entry age, which numerous studies instrument for using laws regulating entry age (\citet{Aliprantis12a, BaruaLang16a}, for an overview, see \citet{FioriniStevens14a}). 
Typically, these laws select an arbitrary date of birth before which children are allowed to enter school in given calendar year. If the cutoff date is January 1, for example, most children born in December will be about 11 months younger when they enter school than children born in January. Due to imperfect compliance with the instrument, however, some fraction of December children may ``redshirt'' and start the following school year, at which time they will be one month \textit{older} than than January children who started on time. Unless the December children who redshirt would also have redshirted had they been born in January, monotonicity is violated. That is, the effect of December birth on school-entry age is typically negative, but for a few children it is positive
(an analogous logic holds for January children who start early).

\begin{figure}
  \centering
  \includegraphics[width=.8\textwidth]{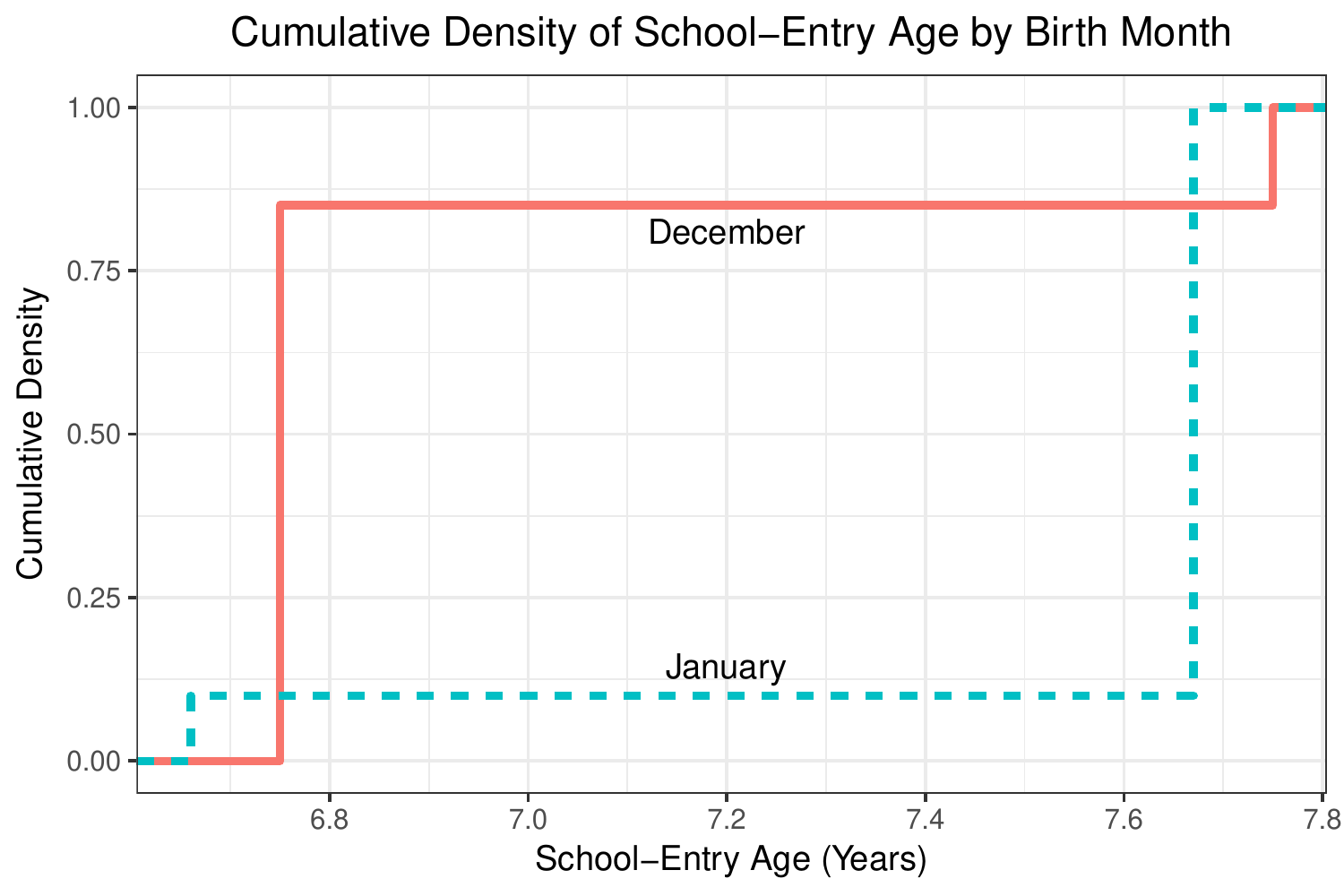}
  \caption{Empirical CDF of school-entry age by birth month (\cite{BlackEtAl11a}).}
  \label{fig:cdf}
\end{figure}

For a simple illustration of how RI can be used to evaluate monotonicity, we re-analyze data on 104,000 children from \poscite{BlackEtAl11a} IV study of school-entry age, which have previously been analyzed from a sampling-based perspective by \citet{FioriniStevens14a}. The latter authors note that the cumulative distribution functions (CDFs) of December- and January-born children in these data cross each other, suggesting a violation of monotonicity.\footnote{
  The crossing could also indicate that the instrument is not exogenous, but here we assume that birth month is as-if randomly assigned. Based on an asymptotic test of stochastic dominance, \citet{FioriniStevens14a} reject the null that the true CDFs do not cross. 
} Figure \ref{fig:cdf} indicates this clearly. Most children started school at an older age if they were born in January rather than December. Indeed, this first-stage relationship is incredibly strong, with an average effect of 0.7 years and an $F$ statistic of 106,146. This would conventionally be considered persuasive evidence of a valid instrument. Note, however, that at the tails of the distribution the relationship between month of birth and entry age reverses: January birthdays predominate among the youngest starters, and December does so among the oldest. 

If we designate January birth as the assigned-to-treatment condition, then the monotonicity assumption is equivalent to the null hypothesis of non-inferiority: being born in January did not cause any child to go to school at an earlier age than they would have if born in December. We test the non-inferiority null using two different effect-increasing test statistics. The first is simply the January--December difference in the proportion of starting ages above 7.67 years (or, equivalently, above the 92.5th percentile of the pooled distribution). The second is the Stephenson rank statistic with subset size 10.\footnote{
  In this application, the performance of Stephenson tests with different subset sizes is essentially given by the redshirting rate among December births, which is 15\%. Specifically, the Stephenson rank statistic totals the number of times a December child will be oldest in a random subset of size $m$, which is equal to the probability that the subset includes either (a) at least one redshirted December child or (b) no on-time January children. Non-inferiority is resoundingly rejected for all subset sizes greater than $m=6$.
} Under either test, the null hypothesis of non-inferiority can be rejected with an exact $p$-value less than $10^{-16}$. We can therefore confidently conclude that being born in January decreased school-entry age for at least some students: the IV monotonicity assumption is violated in this application.\footnote{
  It may still be worthwhile to go forward with the analysis, perhaps conducting a sensitivity analysis for the resulting bias. But it should be noted that redefining the treatment in terms of a dichotomous age cut-off (e.g., below 7.67 years), which would seem to restore monotonicity, would violate the portion of SUTVA that rules out hidden versions of treatment (in this case, different entry ages all defined as ``below 7.67'').
}

\subsection{Inferring Effect Sizes in a Small-Sample Experiment}
\label{sec:inferr-effect-sizes}

For our second application, we move from the relatively easy task of testing hypotheses in a large sample to the more challenging one of inferring effect sizes in a very small one. Specifically, we re-analyze data from \citet{Wantchekon03a}, who convinced Beninese presidential candidates to randomly vary the content of their campaigns across villages.\footnote{
  T
  he village-level data are reported in Table 2 of \citet[412]{Wantchekon03a}.
} The original experiment involved three conditions, but our re-analysis focuses on two: a treatment condition consisting of a purely policy-based campaign and a control condition consisting of a mix of policy and clientelist messages. In each of 8 electoral districts, one village was assigned to each condition, for a total of 16 paired units. Wantchekon's main expectation was that policy campaigns would be relatively ineffective and thus would decrease candidate vote share relative to mixed campaigns. However, given that the same policy proposal may evoke divergent reactions from different groups of voters, it is plausible that the policy campaign had heterogenous effects, a possibility noted by \citet[413]{Wantchekon03a}.\footnote{
  According to the original paper all the mean differences between treatment conditions (and even individual villages) are highly statistically significant, but, as \citet{Green08a} note, \citeauthor{Wantchekon03a}'s analysis ignores village-level clustering and thus vastly overstates the precision of the estimates. 
} Of interest, then, is not only whether policy campaigns were generally less effective but also whether for at least some villages they were \textit{more} effective---and if so, by how much.

\begin{figure}
  \centering
  \includegraphics[width=\textwidth]{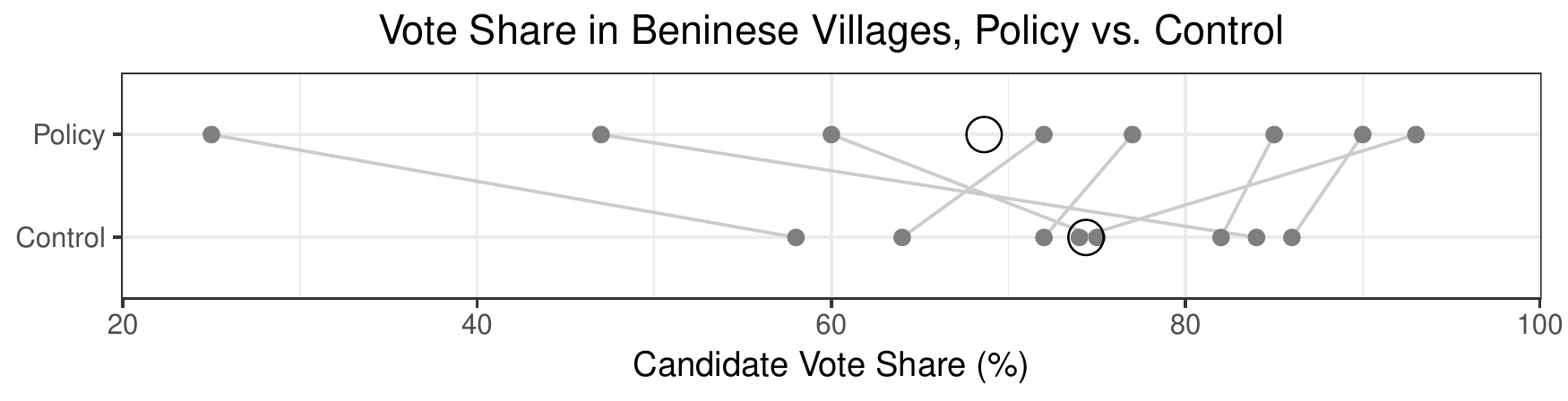}
  \caption{
    Data from the policy and control conditions of \citet{Wantchekon03a}. Hollow circles indicate group means. Villages in the same district are connected by gray lines.
}
  \label{fig:WantData}
\end{figure}

The entire 16-unit dataset is displayed in Figure \ref{fig:WantData}, which plots village-level vote share in each of the two treatment conditions (policy and control). As indicated by the hollow circles, average candidates' vote share in the policy condition was almost 6 percentage points lower than in the control condition, which is consistent with Wantchekon's main hypothesis that policy campaigns were generally less effective. This difference, however, is not statistically significant. A paired $t$ test yields 
a one-sided $p$-value of 0.23 and two-sided 90\% CI that ranges from $-19$ to $+8$ percentage points (see Figure \ref{fig:CI}, top line). Permutation tests using the difference of means or the Wilcoxon rank sum yield very similar two-sided CIs for a constant treatment effect (Figure \ref{fig:CI}, second line). Based on these results, it would appear that the sample size of this experiment is too small to support firm inferences about the magnitude of treatment effects.

We can say more, however, by focusing on extreme effects rather than typical (e.g., average) ones. If treatment effects are heterogeneous, 
then using a test statistic sensitive to the most negative or positive effects can provide greater power. 
Indeed, as  Figure \ref{fig:WantData} indicates, the policy condition contains both the two smallest \textit{and} the two largest observations in the sample. This pattern is consistent with the possibility that despite the negative difference of means, policy campaigns were \emph{more} effective than the control  in at least some villages. We can assess this hypothesis formally using the Stephenson rank test
. Testing the non-superiority null against the alternative of some positive effects, we obtain an exact significance level of $p^{\vee} = 22/256 \approx 0.086$, providing evidence that the policy campaign treatment was indeed more effective in at least some villages.\footnote{
  This is for the Stephenson ranks with subset size $m=6$, but $p$-values for $m=$ 7, 8, 9 and 10 are also around 0.09. Since villages were randomized pairwise within districts, the total number of possible treatment assignments is $2^8=256$.
} There is only slightly weaker  evidence that the policy treatment also had at least one negative effect: a Stephenson rank test of the non-inferiority null gives $p^{\wedge}=32/256=0.125$. By the intersection-union principle \citep{Berger82a}, we can simultaneously reject both the non-inferiority and non-superiority hypotheses at $p=\mathrm{max}(p^{\vee},\; p^{\wedge})=0.125$. There is thus suggestive evidence that treatment had at least one negative \textit{and} at least one positive effect.

\begin{figure}
  \centering
  \includegraphics[width=0.9\textwidth]{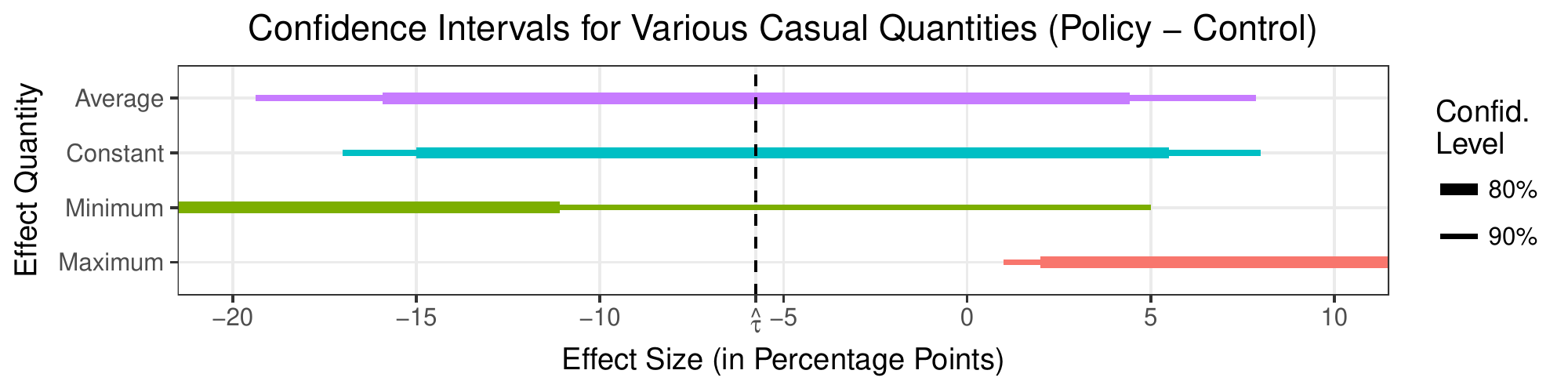}
  \caption{Confidence intervals for various causal quantities in \citet{Wantchekon03a}. The vertical dashed line labeled $\hat{\tau}$ indicates the observed difference of means between the policy and control conditions. The horizontal lines labeled ``Average'' indicate two-sided 80\% and 90\% CIs for the ATE, derived from a paired $t$ test. The lines labeled ``Constant'' are analogous CIs derived from the Wilcoxon rank sum test under the assumption that treatment effects are constant. The lines labeled ``Minumum'' and ``Maximum'' indicate the CIs for the minimum and maximum treatment effects, derived from Stephenson rank tests.}
  \label{fig:CI}
\end{figure}

By inverting the Stephenson rank tests, we can also make inferences about the magnitude of these effects. The lower bound of the one-sided 90\% CI for the maximum effect of the policy treatment is $+$1.0 points, and the lower bound of the 80\% interval is $+$2.0 (Figure \ref{fig:CI}, bottom line). The upper bound on the maximum effect, deduced from vote share's theoretical range, is $+93.0$.\footnote{
  If the outcome variable were unbounded, the outer bounds of the maximum and minimum CIs would be $\infty$ and $-\infty$, respectively. But since the range of the outcome is $[-100, 100]$, we can deduce the maximum possible treatment effect from the observed data.
} The analogous CIs for the minimum effect are, respectively, $[-86.0,\; +5.0]$ and $[-86.0,\; -11.1]$ (Figure \ref{fig:CI}, third line).\footnote{
  The large difference between the 90\% and 80\% bounds for the minimum effect is driven by the large gaps between observations at the low end of the policy group.
}  Note that the 80\% CIs for the minimum and maximum effects do not overlap with one another, which is consistent with the evidence that the policy treatment had both  negative and positive effects. Indeed, the gap between them ($-$11.1 to $+$2.0) suggests that the effects varied quite a bit in magnitude.
 In sum, even in a small sample where mean differences are uninformative, RI enables us to make inferences about treatment effects well beyond the simple conclusion that ``treatment had some effect.'' Moreover, these inference are exact rather than approximate, rely on no assumptions beyond SUTVA and random assignment, and without further assumptions could not be made with other methods.

\section{Conclusion}
\label{sec:conclusion}

The rise of nonparametric causal inference in the tradition of \citet{Neyman23a} and \citet{Rubin:1974wx} has been one of the most important recent developments in quantitative social science. This perspective, with its focus on average treatment effects and its acceptance of effect heterogeneity as the rule rather than the exception, has rightly prompted greater skepticism of statistical methods that heavily rely on parametric models or assumptions. It is then perhaps no surprise that permutation inference, which has traditionally been motivated in terms of shift hypotheses or other highly structured models of treatment effects \citep[e.g.,][]{Lehmann75a, Rosenbaum02a}, has also been regarded with some skepticism despite its freedom from many distributional assumptions.

We have shown that the view of permutation tests now dominant---that despite these tests' virtues, they are useful only for assessing the typically uninteresting and implausible sharp hypothesis that treatment had no effect at all---is too limited. We have proved that permutation tests using effect-increasing test statistics are valid under a more general bounded null, that this fact can be exploited to derive confidence intervals for the maximum effect, and that many familiar test statistics are effect-increasing. We have also highlighted the value of less familiar statistics such as the Stephenson rank sum, which is sensitive to the extremes of the treatment effect distribution, and explained the normative and theoretical relevance of bounded null hypotheses. Finally, we have re-analyzed a well-known experiment to demonstrate that, when coupled with the new interpretation we have advanced, permutation tests can yield substantively interesting inferences about treatment effects that are not possible based on ATE estimation alone.

In sum, we have developed a novel perspective on permutation tests that we hope tempers the skepticism that many social scientists hold towards this otherwise-appealing mode of statistical inference. Permutation tests are by no means a cure-all; nor are they a substitute for ATE estimation when that is the goal of the analysis. But in many cases, particularly when samples are small, treatment groups unequal, or treatment assignment complex, they are the most reliable form of statistical inference. Moreover, even when this is not the case, they often make possible inferences about treatment effects that other methods cannot. For these reasons, permutation tests deserve a secure place in the quantitative social scientists' toolbox.

\bibliographystyle{abbrvnat}
\bibliography{permutation}


\newpage
\appendix

\section{Supplementary Appendix}
\label{sec:online-appendix}

\subsection{The $t$ Test with Skewness and Unequal Sample Sizes}
\label{sec:performance-t-test}

The following code gives rejection rates of 17\% against the null with a nominal $\alpha$ of 5\%, demonstrating how the $t$ test can fail even in moderate-sized samples if there is a strong skew.

\begin{lstlisting}
  ### DGP with one small group and one large, and skew.
  n1 <- 30
  n2 <- 1000
  alpha <- .2
  beta <- 20
  ## Plot density
  plot(seq(0, 1, 0.01), dbeta(seq(0,1,0.01), alpha, beta), type="l")
  ## Simulations
  sims <- 10000
  ps <- rep(NA, sims)
  set.seed(1)
  for (i in 1:sims){
    ys <- rbeta((n1 + n2), 0.1, 5) ## beta
    zs <- sample(c(rep(1, n1), rep(0, n2)), replace=FALSE)
    ps[i] <- as.numeric(t.test(ys ~ zs, var.equal=FALSE)[3]) ## Welch/Neyman
  }
  ps <- ps[!is.na(ps)]
  length(ps)
  mean(ps < 0.01)
  mean(ps < 0.05)
\end{lstlisting}

\subsection{Proof of Validity under the Bounded Null}
\label{sec:proof-validity-under}

\begin{proof}
  Let $\sch_h \in \boundedH$ be any potential-outcomes schedule satisfying the
  non-superiority null hypothesis
  $\tau_i \leq \taumax_{i} \equiv \tau^0_{i}\ \forall i \in 1\ldots n$.
  Suppose $\sch_h$ holds. We randomize the units and obtain $\wobsV$,
  $\yobsV$, and $\tobs = T( \wobsV, \sch_h )$.  We then impute
  $\schimp_{\vtau^0} = \schimp( \wobsV, \yobsV, \sharpH )$ and obtain $\pimp$,
  our nominal $p$-value.  Even though $\nonsupH$ holds, $\sharpH$ might not,
  and so it is possible that $\schimp_{\vtau^0} \neq \sch_h$ and thus
  $\pimp \neq p$.  However, if the test statistic is effect-increasing, then
  we can place a bound on $\pimp$. In particular, note that
  \[
    \Yimp_i(1) =
    \begin{cases}
      Y_i(1) & \wobs_i = 1 \\
      Y_i(0) + \tau^0_i & \text{otherwise}
    \end{cases}
  \]
  for all $i \in 1\ldots n$. Under $\nonsupH$
  \[
    \tau^0_i \ge \tau_i = Y_i(1) - Y_i(0),
  \] so \[ Y_i(1) \leq Y_i(0) + \tau^0_i
  \]
  and thus
  \[
    Y_i(1) \leq \Yimp_i(1).
  \]
  In other words, every true potential outcome under treatment is no larger
  than its imputed equivalent.  By analogous logic $Y_i(0) \geq
  \Yimp_i(0)$. These inequalities imply $\sch_h \preceq
  \schimp_{\vtau^0}$. Since $T(W, \sch)$ is effect-increasing,
  $T(\vw, \sch_h) \leq T(\vw, \schimp_{\vtau^0})$ for any realization
  $W = \vw$. In other words, because the potential-outcome schedules are
  ordered $\sch_h \preceq \schimp_{\vtau^0}$, the values of $T$ simulated from
  $\schimp_{\vtau^0}$ will be pointwise weakly larger than $T$'s true
  reference distribution. As a consequence,
  \[
    \pimp = \pr{ T( \vw^*, \schimp_{\vtau^0} ) \geq \tobs } \geq \pr{ T(
      \vw^*, \sch_h ) \geq \tobs } = p,
  \]
  i.e., the estimated $p$-value will be at least as large as the true
  one. This gives a valid (though potentially conservative) test:
  \[
    \pr{ \pimp \leq \alpha } \leq \pr{ p \leq \alpha } \leq \alpha.
  \]
\end{proof}

\subsection{Proof of Validity of CI for Maximum Effect}
\label{sec:proof-validity-ci}

In this section, we prove that conventional confidence intervals for a constant treatment effect are valid nonparametric intervals for the maximum effect. To show this, we first need a small lemma:
\begin{lemma}
For our one-sided testing case, and regardless of the character of the true $\sch$, the one-sided CI for an effect-increasing statistic will be a half-interval $[L,\infty)$, indicating that the constant-shift treatment effect is no smaller than $L$.
\end{lemma}
\begin{proof}
Say our CI is not a half-interval. Then there exists $\tau_1 < \tau_2$ such that $\tau_1$ is not in CI and $\tau_2$ is. But the proof of our main result shows that if we are testing $\rH_{\tau^0_2}$ then the $p$-value will be lower for any $\sch' \preceq \schimp$, including the one corresponding to a constant treatment effect shift of $\tau_1 < \tau_2$.
But this means $\tau_1$ would not be in CI, which is a contradiction.
Therefore the CI is a half-interval.
\end{proof}

These confidence intervals can easily generalize to bounded nulls. 
Let CI be the above confidence set generated by inverting a sequence of constant-effect nulls. Now consider the true potential-outcomes schedule $\sch$.
Assuming all potential outcomes are defined,  let
\[ \tau^* \equiv \arg\max_i \tau_i \]
be the largest treatment effect in schedule $\sch$.
Then $\nonsupH: \tau_i \le \tau^{*}\ \forall i$ is true, and so testing the associated $\sharpH: \tau_i = \tau^{*}\ \forall i$ will reject with probability no greater than $\alpha$. We therefore will include $\tau^* \in \mathrm{CI}$ with probability no less than $1-\alpha$, giving a valid CI for the maximum effect.

\subsection{Effect-Increasing Statistics}
\label{sec:proof-that-rank}

Many statistics are effect-increasing. In this appendix we first prove that differences in rescaled outcomes are effect-increasing.
We then show that a general class of rank-based statistics is EI in the absence of ties, and then extend this results to the case of ties for, the rank-sum and Stephenson rank-sum statistics, two primary statistics discussed in the paper.

\begin{proof}[Lemma 2]
Let $Q_i = Q(Y_i)$ denote the observed scores, and let $Q_i(1) = Q(Y_i(1))$ and $Q_i(0) = Q(Y_i(0))$ indicate the scores of the potential outcomes. 
Denote as $S(W, \sch)$ any statistic with the form $a(W) \sum_i W_i Q(1)_i = a(W) \sum_i W_i Q_i$ with $a(W) \geq 0$, i.e., the scaled sum of the scores of the treated observations.\footnote{
  \label{fn:sumstat}
  Note that as $Q(Y_i)$ depends only on unit $i$, not on the entire vector $Y$. This is a crucial difference from the class of sum statistics defined by \citet[35]{Rosenbaum02a}. As a result of this distinction, the class of statistics $S(W, \sch)$ defined here excludes rank statistics because the rank of $Y_i$ depends on the values of other observations. We treat ranks statistics separately below.
} 
We first show that for any $S(W, \sch)$ and any pair of potential-outcome schedules $\sch \preceq \sch'$, $S(W, \sch) \le S(W, \sch')$:
\begin{align*}
  Y_i(1) \le Y'_i(1)\ \forall i & \implies Q_i(1) \le Q'_i(1) \ \forall i &
              \text{b/c\ } Q(Y_i) \text{\ is non-decreasing in}\ Y_i\\
  & \implies  Q_i(1) - Q'_i(1) \le 0 \ \forall i\\
  & \implies  \sum W_i [Q_i(1) - Q'_i(1)] \le 0 
             & \text{b/c every element is} \le 0 \\
  & \implies  \sum [ W_iQ_i(1)] -  \sum [W_iQ'_i(1)] \le 0 \\
  & \implies  a(W) \sum [ W_iQ_i(1)] \le a(W) \sum [W_iQ'_i(1)] \\
  & \implies  S(W, \sch) \le S(W, \sch') . 
\end{align*}
The above includes the special case of the sum of the treated responses, $\sum_i W_i Y(1)_i$, for which $Q(Y_i) = Y_i$ and $a(W)=1$. 
The intuition behind the above is that raising any individual potential outcome on the treatment side will either (if the unit was treated) increase $\sum_i W_i Y(1)_i$ or (if the unit was not treated) not affect the test statistic at all.

This extends to differences the scaled sum of scores as in the lemma because the scaled negative of the sum of the control scores, $-b(W) \sum [ W_iQ_i(0)]$, is effect increasing, and the sum of two effect increasing statistics is also effect increasing.

In particular, by letting $Q(Y_i) = Y_i$, $a(W)=1/n_T$, $b(W)=1/n_C$ we have our difference in means as being effect increasing.
This also applies to assignment mechanisms with variable numbers of treated units (once we define $a(W)$ and $b(W)$ for those $W$ that have no treated or no control units).


\end{proof}

\begin{proof}[Lemma 3]
Let the $R_i$ be the ranks of the (adjusted) observed outcomes, $Q(\cdot)$ a monotonic mapping of these ranks to the real numbers, and $a(W), b(W)$ nonnegative functions of the assignment vector.

Consider two schedules $\sch' \preceq \sch$ that are identical except that for some specific unit $k$ with $Y_k'(1) \leq Y_k(1)$.
Conceptually consider making $\sch'$ by reducing potential outcome $Y_k(1)$ for some specific $k$, and leaving the other potential outcomes alone.

Now given any assignment vector $W$, we have $Y^{obs}_i, i = 1, \ldots, n$, and $R_i, i = 1, \ldots, n$ the associated ranks.
Assume no ties in ranks.
By lowering $Y_k(1)$ to $Y_k'(1)$ we potentially could change some ranks.  
In particular, the rank of unit $k$ could go down and if it does the ranks of some other units would increase.
Let $G = \left\{ j : R_j' > R_j \mbox{ and } W_j = 1 \right\}$ be the set of units in the treatment group with increased ranks.
Let $m$ denote the size of this set.
Let $j_1, \ldots, j_m$ be the indices of the units in $G$ arranged in increasing order by rank, so $R_{j_a} < R_{j_{a+1}}$.
Importantly, for any unit in $G$, we have a change of at most 1 rank, so $R_{j_{a+1}} \geq R_{j_{a}} + 1 \geq R'_{j_{a}}$, giving $Q( R'_{j_{a}} ) \leq Q( R_{j_{a+1}} )$.
Because no unit can increase its rank to above $R_k$ by reducing unit $k$ we have $R'_{j_m} \leq R_k$.
Similarly, the reduced $R_k'$ must be less than the rank of any unit in $G$, giving $R'_k \leq R_{j_1}$.
This gives:
\[ Q( R'_{k} ) + Q( R'_{j_1} ) + \ldots + Q( R'_{j_{m-1}} ) + Q( R'_{j_m} ) \leq Q( R_{j_1} ) + Q( R_{j_2} ) + \ldots + Q( R_{j_m} ) + Q( R_k ) \]
due to a pairwise comparison (the first elements of the two sums are ordered, the second, etc., up to the $m$th).

This means that $T' \leq T$ because the treatment sum decreases with the change, and as the control average can only go up (those units impacted in the control group all have ranks that are larger than they were previously), subtracting this second term also decreases the total.
As with the prior proof, the scaling quantities $a(W)$ and $b(W)$ are unchanged, and therefore carry through.

A similar argument shows that $T' \leq T$ if we consider a pair of potential-outcome schedules where only a single control potential outcome is increased from $\sch$ to $\sch'$.

Finally, take any two potential-outcome schedules $\sch' \preceq \sch$.  Generate a chain of potential-outcome schedules from $\sch'$ to $\sch$ by changing one potential outcome at a time.  
For example, the first step in the chain would be to modify $\sch'$ to $\sch''$ so $Y_1(1)'' = Y_1(1)$ and all other $Y_i''(z) = Y_i'(z)$.
By transitivity along this sequence we finally have $T' \leq T$ for any $W$.
Therefore, $T$ is potential outcomes monotonic.
\end{proof}

\paragraph{Examples.}
If $Q(r) = r$ we have the classic rank sum test.
Similarly, if 
\[ Q(r) = {r - 1 \choose s - 1} \mbox{ for } r \geq s \mbox{ and } Q(r) = 0 \mbox{ otherwise } \]
for some fixed $s$ (representing how many subsets of size $s$ can be formed where our unit with rank $r$ is biggest), we have the Stephenson Rank test.

\paragraph{Ties.} Unfortunately, this general proof does not go through if there are ties.  
This is because by lowering a potential outcome, an entire group of mid-ranks can shift.
Consider the case where the original ranks of treated units are $2, 2, 2, 8$ and we lower the rank-$8$ so much that it becomes rank 1.
The other units then will have ranks $3, 3, 3$ giving final ranks of $1, 3, 3, 3$.  
Now, if $Q(r) = 0$ for $r < 3$ and 1 otherwise, the sum of the four goes from 1 to 3.  
The control units are unaffected.
This violates the monotonicity.
Many specific rank based statistics are, however, EI even in the presence of ties.
This can be shown by direct proof. 
We next do this for the rank-sum and the Stephenson rank-sum statistics.

\paragraph{The rank-sum test.}
Let 
\[ T( W, \sch, H_{s\delta} ) = \sum W_i \text{rank}( \Yobs_i - \delta \Wobs_i ) \]
This statistic is equivalent to the Mann-Whitney statistic summing all pairs of treatment-control observations with the treatment beating the control
\[ T_{MW}( W, \sch ) = \sum_{i, j} W_i (1-W_j) \ind{ \Yobs_i - \delta \geq \Yobs_j } = \sum_{i, j} W_i (1-W_j) \ind{ Y_i(1) - \delta  \geq Y_i(0) } \]
with $\ind{a \leq b}$ equalling 1/2 if $a=b$.

Then, for any two potential-outcome schedules with $\sch \preceq \sch'$ we have
\[ \ind{ Y_i(1) - \delta  \geq Y_i(0) } \leq \ind{ Y'_i(1) - \delta \geq Y'_i(0) } \]
since we are moving the left side up and the right side down, which only increases the chance of having the left side be higher.
Plugging this in to our sum of pairwise comparisons easily obtains our result of $T_{MW}( W, \sch ) \leq T_{MW}( W, \sch' )$.

\paragraph{The Stephenson rank test.}
Represent this statistic as a sum of indicators across all subsets where the indicator is 1 if a treatment unit is (tied for) the largest.  
We have, letting $G$ indicate a size-$s$ subset of unit indices and $\mathcal{G}$ the collection of all such $G$,
\[ T_{S}( W, \sch ) = \sum_{G \in \mathcal{G} } H_G \]
where
\[ H_G = \max_{i \in G} W_i \mbox{ s.t. } \tilde{Y}_i \geq \max_{j \in G} \tilde{Y}_j \]
with $\tilde{Y}_i$ being an adjusted outcome (i.e., imputed control outcome under the null).
The above simply says that $H_G$ is 1 if there is a treated unit that is (tied for) largest value in the set $G$.
Alternatively, substitute $Y_i^{obs}$ for $\tilde{Y}_i$.

Then, for any two potential-outcome schedules with $\sch \preceq \sch'$ and a given $G$ we have $H_G$ and $H_G'$ with
\[ H_G \leq H_G' \]
since for any unit under treatment, $Y_i^{obs}$ can only be larger, and for control, smaller.  
Therefore, for each subset where a treatment unit was largest for $\sch$, we will still see one being largest for $\sch'$.
These inequalities sum, giving $T_S \leq T_S'$ for any $W$, which implies monotonicity.

\end{document}